\documentclass[a4paper]{scrartcl}

\usepackage{amsmath}

\usepackage[T1]{fontenc}

\usepackage{cite}
\usepackage{amssymb,amsfonts}
\usepackage{thm-restate}
\usepackage{algorithm}
\usepackage{algpseudocode}
\usepackage{subcaption}
\usepackage{pgfplots}
\usepackage{graphicx}
\usepackage{textcomp}
\usepackage{xcolor}
\usepackage{hyperref}
\usepackage{wrapfig}
\usepackage{xifthen}
\def\BibTeX{{\rm B\kern-.05em{\sc i\kern-.025em b}\kern-.08em
    T\kern-.1667em\lower.7ex\hbox{E}\kern-.125emX}}

\newcommand{\opers}{\Omega}

\newcommand{\Nat}{\mathbb{N}}

\newcommand{\term}{\mathcal{T}}
\newcommand{\varset}{\mathcal{X}}

\newcommand{\gnd}{\operatorname{gnd}}

\newcommand{\lic}{\operatorname{lic}}
\newcommand{\size}{\operatorname{size}}
\newcommand{\vars}{\operatorname{vars}}
\newcommand{\mgu}{\operatorname{mgu}}
\newcommand{\dom}{\operatorname{dom}}
\newcommand{\cdom}{\operatorname{cdom}}

\newcommand{\shortrules}[6]{\noindent\begin{minipage}{#6ex}{\bfseries #1}\end{minipage} $\;$ #2 $\;\Rightarrow_{\text{#5}}\;$ #3 \ifthenelse{\isempty{#4}}{}{\par\noindent #4}}

\newcommand{\CCVAR}{{\text{CC($\mathcal{X}$)}}}
\newcommand{\CC}{$\text{CC}$}

\usepackage{amsthm}
\newtheorem{definition}{Definition}

\newtheorem{lemma}[definition]{Lemma}
\newtheorem{corollary}[definition]{Corollary}
\newtheorem{example}[definition]{Example}

\theoremstyle{remark}

\begin{document}
\title{Non-Ground Congruence Closure}
%
%

\author{Hendrik Leidinger,
  Christoph Weidenbach\\ 
  \large{Max Planck Institute for Informatics} \\ \large{Saarland Informatics Campus} \\ \large{Saarbr\"ucken, Germany} \\ \large{\texttt{\{hleiding,weidenb\}@mpi-inf.mpg.de}}}

%

\maketitle

\begin{abstract}
    Congruence closure on ground equations is a well-established and efficient algorithm for deciding ground equalities. 
	It constructs an explicit representation of ground equivalence classes based on a given set of input equations, 
	allowing ground equalities to be decided by membership. In many applications, these ground equations originate from grounding non-ground equations.

We propose an algorithm that directly computes a non-ground representation of ground congruence classes for non-ground equations. 
Our approach is sound and complete with respect to the corresponding ground congruence classes. Experimental results demonstrate 
that computing non-ground congruence classes often outperforms the classical ground congruence closure algorithm in efficiency.
\end{abstract}
\section{Introduction}
\label{sec:intro}


Equational logic is widely used in almost all aspects of formal reasoning about systems.
Given a set of non-ground equations, in order to apply congruence closure (\CC)~\cite{DowneySethiEtAl80}
the non-ground equations must be grounded first. Then after applying \CC\  to the grounded equations, ground equalities can be decided by
testing membership in the generated congruence classes. \CC\ is in particular useful compared to ground completion~\cite{KnuthBendix70},
if many different ground equalities are tested, in particular, with respect to the same set of input equations. For \CC\ testing an equality amounts to a membership test, for completion
normal forms need to be computed.
If the starting point for reasoning is a set of first-order equations with variables, then applying \CC\ 
amounts to grounding the set of equations by a set $\mathcal{M}$ of ground terms
resulting in an exponential blow-up in $|\mathcal{M}|$.
In this paper we introduce \CCVAR\ that prevents the a priori exponential blow-up by grounding
and operates on the original equations. Still \CCVAR\ ensures that any congruence class it derives
has ground instances in $\mathcal{M}$.

In \CCVAR\  a (congruence) class is a set of \emph{constrained terms}, also called a 
\emph{constrained class}. A constraint is a conjunction of checks resembling membership 
conditions and of the form $t\in\mathcal{M}$, where $t$ is a possibly non-ground term 
and $\mathcal{M}$ is a set of ground terms. When presenting  \CCVAR\  we assume $\mathcal{M}$ to be finite.
For example, such a set of ground terms can either explicitly represented by enumerating the considered ground
terms (as done in SMT) or implicitly through an ordering that always guarantees only finitely many smaller ground
terms given some upper bound ground term.
The \emph{constraint} restricts the number of ground instances of the \emph{constrained term} to those in $\mathcal{M}$. 
\CCVAR\ creates a solution for the whole ground input space 
$\mathcal{M}$. It consists of two main rules \emph{Merge} and \emph{Deduction} to build up
the congruence classes that are generalizations of the Merge and Deduction rule from \CC.
\emph{Merge} creates a new class by unifying two terms in different classes and applying this 
unifier to the Union of these two classes. \emph{Deduction} creates a new class by simultaneously unifying the arguments of 
two terms with the same top symbol in different classes with terms in the same class. Termination is guaranteed 
by an effective  subsumption test between two classes and the fact that all terms are constrained by $\mathcal{M}$.

As an example, consider the equations $g(x)\approx h(x)$, $h((h(y)))\approx f((h(y)))$ and 
\(\mathcal{M}= \{g(x),\allowbreak h(x), \allowbreak f(x), \allowbreak f((h(x))),\allowbreak h(h(x)),\allowbreak g(h(x)),
\allowbreak g(g(x)),\allowbreak h(g(x)), \allowbreak f(g(x))~|~\allowbreak x\in \{t_1,...,t_n\},\allowbreak\;\text{$t_i$ ground}\}\). Initially \CCVAR\ creates the classes:
\smallskip

\centerline{$\begin{array}{c}
\{ g(x), h(x)\parallel g(x), h(x)\}, \{h(h(y)),f(h(y))\parallel h(h(y)),f(h(y))\}
               \end{array}$}
\smallskip

\noindent where the left-hand side of the $\parallel$ operator is the set of constraints and the right-hand side is the set of terms.
The set $\mathcal{M}$ of ground terms is fixed, therefore for readability we simply write, e.g., $\{g(x),h(x)\parallel g(x),h(x)\}$ instead of  $\{g(x)\in\mathcal{M},h(x)\in\mathcal{M}\parallel g(x),h(x)\}$
and use the notation $\{\Gamma\parallel s_1\ldots,s_n \}$ to denote that the 
constraint $\Gamma$ holds for all terms. \CCVAR~can now merge
the classes $\{g(x),h(x)\parallel g(x),h(x)\}$ and $\{h(h(y)),f(h(y))\parallel h(h(y)),f(h(y))\}$ by unifying $h(x)$
with $h(h(y))$.
So we get $\{g(h(y)),h(h(y)), f(h(y))\parallel g(h(y)),h(h(y)), f(h(y))\}$. The new class subsumes the second merged class, so 
it can be deleted. Now we can apply three times Deduction with one of the implicitly available single term classes $\{f(y)\parallel f(y)\}$,
$\{g(y)\parallel g(y)\}$ and $\{h(y)\parallel h(y)\}$ and the class $\{g(x),h(x)\parallel g(x),h(x)\}$
to create the new class $\{f(g(x)), f(h(x))\parallel f(g(x)), f(h(x))\}$. Repeating this procedure for the other two single-term classes
we get $\{g(g(x)), g(h(x))\parallel g(g(x)), g(h(x))\}$ and $\{h(g(x)), h(h(x))\parallel h(g(x)),\allowbreak h(h(x))\}$. 
All these classes can now be merged step by step with the first merged class, whereby the involved classes are always subsumed by the 
new class. The final result is now:
\smallskip

\centerline{$\begin{array}{c}
\{g(x),h(x)\parallel g(x),h(x)\},\\
\{g(h(y)), g(g(y)), h(h(y)), h(g(y)), f(g(y)), f(h(y))\parallel\\ 
g(h(y)), g(g(y)), h(h(y)), h(g(y)), f(g(y)), f(h(y))\}
             \end{array}$}
\smallskip

Comparing the application of \CCVAR\ with \CC, grounding with $\mathcal{M}$ leads to multiple copies of the involved classes if the number of ground terms is at least higher than one
and even if the ground terms $t_i$ do not contain the function symbols $f$, $g$, $h$.
For example, if the $t_i$ represent 8 different constants, ground congruence closure would create a total number of 16 classes, whereas in \CCVAR\
the number of classes is independent of the number of constants in this example. In Section~\ref{sec:evaluation} we will show that \CCVAR\ typically
outperforms \CC.

\paragraph{Related Work:} 
To the best of our knowledge, the only algorithm that is similar to ours is Joe Hurd's
Congruence Classes with Logic Variables~\cite{Hurd01}. The algorithm creates a set of classes, where 
each class consists of multiple, possibly non-ground terms. It incrementally finds all 
matchers between all pairs of classes and applies these matchers to extend these classes.
Therefore, in order to test equality of two terms they need to be added and the algorithm restarted.
The size of terms is not constrained by the algorithm so it may diverge. There is no notion
of redundancy. The semantics of the classes is with respect to an infinite signature. In contrast,
we generate a complete classification of all considered ground terms, i.e., testing equality of two
terms means testing membership in the same class. Due to a notion of redundancy, our algorithm always terminates (Lemma~\ref{lem:cccterm}),
and it is implemented (Section~\ref{sec:implementation}).

There are two potential applications of \CCVAR\ in existing frameworks: SMT (Satisfiability Modulo Theory) and SCL(EQ) (Simple Clause Learning in first-order logic with equality).
Satisfiability modulo theory (SMT)~\cite{GanzingerHagenEtAl04,NieuwenhuisEtAl06} solvers (e.g., \cite{deMouraBjorner:08,BarbosaBBKLMMMN22,BoutonODF09,Dutertre2014,CimattiEtAl13}) make
use of congruence closure (\CC)~\cite{NelsonOppen80,DowneySethiEtAl80,Shostak84}.
It is well suited as an incremental algorithm that decides the validity of ground equations and is 
able to ﬁnd a small subset of input equations that serves as a proof, all in time $\mathcal{O}(n~log (n))$~\cite{nieuwenhuis2007fast}.
Typical applications consider a set of first-order equations with variables plus further theories.
For SMT solvers many techniques have been invented to instantiate non-ground input equations~\cite{BarbosaEtAl17,ReynoldsEtAl18, ReynoldsTinelli14} in order to make them applicable to \CC.
With \CCVAR\ the grounding could be postponed allowing reasoning modulo a larger set of grounding terms.

SCL(EQ) (Simple Clause Learning over Equations)~\cite{Leidinger2023}  is a complete calculus for first-order logic with equality that only learns 
non-redundant clauses. Similar to CDCL (Conflict Driven Clause Learning)~\cite{MSS96,BayardoSchrag96,MoskewiczMadiganZhaoZhangMalik01,BiereEtAl09handbook,Weidenbach15}
it builds an explicit model assumption (trail) of ground literals where ground instances are finitely limited by some ground term $\beta$ with respect to an ordering.
Clauses are evaluated with respect to a trail. \CCVAR\ can be immediately applied by not grounding the equations but reasoning modulo a set of ground
terms smaller than $\beta$.

For uniform word problems~\cite{DershowitzPlaisted01handbook} CC(X) can be used as an alternative semi-decision
procedure by iteratively increasing the set of considered ground terms.
It may also provide a priori guarantees, e.g., if there is a depth
bound on the overall set of ground terms that need to be considered.
Whenever the overall set of ground terms to be considered is finite CC(X) turns into
a decision procedure due to soundness, Lemma~\ref{lem:cccsound}, completeness, Lemma~\ref{lem:ccccomp}, and termination, Lemma~\ref{lem:cccterm}.

The rest of the paper is organized as follows. Section~\ref{sec:prelim} provides the necessary technical
background notions needed for the development of our algorithm. Section~\ref{sec:encoding} presents our new 
calculus including examples and proofs for soundness, Lemma~\ref{lem:cccsound}, completeness, Lemma~\ref{lem:ccccomp}, and termination, Lemma~\ref{lem:cccterm}.
Section~\ref{sec:implementation} presents
refinements of the calculus towards an implementation.
Section~\ref{sec:evaluation} compares the performance of \CCVAR\  with \CC\ by experiments.
All experiments can be reproduced by the provided supplementary material.
Section~\ref{sec:discussion} concludes the paper. An appendix contains additional examples, 
proofs and pseudo code of our implementation.

\section{Preliminaries}
\label{sec:prelim}

We assume a standard first-order language with equality and over a finite set of function symbols $\opers$, 
where the only predicate symbol is equality $\approx$.  We assume that $\opers$ contains at least one constant and
one non-constant function symbol.
$t,s,l,r$ are terms from $\term(\opers,\varset)$ for an infinite set of variables $\varset$;
$f,g,h$ function symbols from $\opers$; $a,b,c$ constants from $\opers$ and $x,y,z$ variables from $\varset$.
The function $\vars$ returns all variables of a term. The function $\#(x,t)$ returns the number of 
occurrences of a variable $x$ in $t$. 

By $\sigma,\tau,\delta,\mu$ we denote substitutions. Let $\sigma$ be a substitution, then its finite domain is defined as
$\dom(\sigma) := \{x \mid x\sigma \not= x\}$ and its codomain is defined as $\cdom(\sigma) = \{t \mid x\sigma = t, x \in \dom(\sigma)\}$. We extend their application to terms and sets of terms in the usual way.
A term, equation is \emph{ground} if it does not contain any variable.
A substitution $\sigma$ is \emph{ground} if $\cdom(\sigma)$ is ground. A substitution $\sigma$ is \emph{grounding} for a term $t$, equation $s\approx t$ if $t\sigma$, $(s\approx t)\sigma$ is ground,
respectively. 
The function $\gnd$ returns the set of all ground
instances of a term, equation or sets thereof.
The function $\mgu$ denotes the most general unifier of two terms, two equations, respectively. We assume that mgus do not introduce fresh variables and that they are idempotent.
The size of a term $t$ (equation $E$) is denoted by $\size(t)$ ($\size(E)$), which is the number
of symbols in $t$ ($E$).

Let $\preceq$ be a total quasi-ordering on ground terms where the strict part is well-founded.
The ordering is lifted to the non-ground case via instantiation: we define $t \preceq s$
if for all grounding substitutions $\sigma$ it holds $t\sigma \preceq s\sigma$.
Given a ground term $\beta$ then $\gnd_{\preceq\beta}$
computes the set of all ground instances of a term, equation, or sets thereof where all ground terms  are smaller or equal to
$\beta$ with respect to $\preceq$. 
Analogously, for a set of ground terms $\mathcal{M}$, we define $\gnd_{\mathcal{M}}$ as the set of ground instances of a term, equation or sets thereof where all ground terms
are in $\mathcal{M}$.
By $\term_{\preceq\beta}(\Omega, \emptyset)$ or just $\term_{\preceq\beta}$
we denote the set of all ground terms $\preceq\beta$.

We rely on standard first-order semantics and in particular write $E\models s\approx t$ if any model of the implicitly
universally quantified equations in $E$ is also a model for the ground equation $s\approx t$.

Let $E$ be a set of equations over $T(\Omega,\mathcal{X})$ where all
variables are implicitly universally quantified. The well-known inference
system of equational logic comprises the following rules~\cite{BaaderNipkow98}

\smallskip

\shortrules{Reflexivity}
{$E$}
{$E\cup \{t\approx t\}$}
{provided $t$ is a term.}{EQ}{12}

\smallskip

\shortrules{Symmetry}
{$E\cup \{t\approx t'\}$}
{$E\cup \{t\approx t', t'\approx t\}$}
{}{EQ}{12}

\smallskip

\shortrules{Transitivity}
{$E\cup \{t\approx t', t'\approx t''\}$}
{$E\cup \{t\approx t', t'\approx t'', t\approx t''\}$}
{}{EQ}{12}

\smallskip

\shortrules{Congruence}
{$E\cup \{t_1\approx t'_1,..., t_n\approx t'_n\}$}
{$E\cup \{t_1\approx t'_1,..., t_n\approx t'_n, f(t_1,...,t_n)\approx f(t'_1,...,t'_n)\}$}
{}{EQ}{13}

\smallskip

\shortrules{Instance}
{$E\cup \{t\approx t'\}$}
{$E\cup \{t\approx t, t\sigma\approx t'\sigma\}$}
{provided $\sigma$ is a substitution.}{EQ}{10}

\smallskip\noindent
and by Bikhoff's Theorem~\cite{BaaderNipkow98} we get for two ground terms $t,s$ : $E\models s\approx t$ iff $E\Rightarrow_{\text{EQ}}^* \{t,s\}\cup E'$.
We will base our completeness proof, Lemma~\ref{lem:ccccomp}, on $\Rightarrow_{\text{EQ}}$.


Congruence Closure~\cite{NelsonOppen80,DowneySethiEtAl80,Shostak84} is an algorithm for deciding satisfiability of ground equations. 
We present the abstract version of congruence closure as described by 
Fontaine\cite{fontaine2004techniques} and based on the algorithm by Nelson and Oppen~\cite{NelsonOppen80}.
The initial state is $(\Pi, E)$, where $\Pi$ is a partition of all ground subterms of terms in $E$, such that every term is in its
own class, and $E$ is the set of ground equations. The algorithm consists of the following three inference rules.

\smallskip

\shortrules{Delete}
{$(\{A\}\cup \Pi, E\cup \{s\approx t\})$}
{$(\{A\}\cup \Pi, E)$}
{provided $\{s,t\}\subseteq A$.}{CC}{6}

\smallskip

\shortrules{Merge}
{$(\{A,B\}\cup \Pi, E\cup \{s\approx t\})$}
{$(\{A\cup B\}\cup \Pi, E)$}
{provided $s\in A$, $t\in B$ and $A\not= B$.}{CC}{6}

\smallskip

\shortrules{Deduction}
{$(\{A,B\}\cup \Pi, E)$}
{$(\{A, B\}\cup \Pi, E\cup\{f(s_1,...,s_n)\approx f(t_1,...,t_n)\})$}
{provided $f(s_1,...,s_n)\in A$, $f(t_1,...,t_n)\in B$, $A\not= B$
 and for each $i$, there exists a $D_i\in \{A,B\}\cup \Pi$ such that $\{s_i,t_i\}\in D_i$
 and $f(s_1,...,s_n)\approx f(t_1,...,t_n)\not\in E$.}{CC}{11}

 \smallskip

The algorithm terminates if no rule is applicable anymore. The resulting set $\Pi$ represents the set
of congruence classes.

\section{CC($\mathcal{X}$): Non-Ground Congruence Closure}
\label{sec:encoding}

We now present our calculus in full detail. 
Throughout this section we assume a finite set $\mathcal{M}$ of ground terms and a finite set $E$ of
possibly non-ground equations.


\begin{definition}
  A \emph{constrained term} $\Gamma\parallel s$ is a term $s$ with a constraint $\Gamma$. 
  The \emph{constraint} $\Gamma$ is a conjunction of atoms $t\in \mathcal{M}$.
  A substitution $\sigma$ is \emph{grounding} for a constraint term $\Gamma\parallel s$ if $\Gamma\sigma$ and $s\sigma$ are ground.
\end{definition}

The constraint $\Gamma$ restricts the possible ground instances of the term $s$ to those instances $s\sigma$ such that 
$\Gamma\sigma$ evaluates to true. If it is clear from the context, we omit the $\in \mathcal{M}$ and just
write the left hand-side of the operation. A constraint class is a set of constraint terms. We distinguish 
between separating and free variables, where a separating variable occurs in all terms within the class whereas
a free variable does not.

\begin{definition}[Congruence Class]
  A \emph{congruence class} or simply class is a finite set of constraint terms $\Gamma\parallel s$.
  Let $A=\{\Gamma_1\parallel s_1,...,\Gamma_n\parallel s_n\}$ be a class. The set of \emph{separating} variables $X$ of 
  $A$ is defined as $X = \vars(s_1)\cap...\cap \vars(s_n)$.
  The set of \emph{free} variables $Y$ of $A$ is defined as $Y = (vars(s_1)\cup...\cup vars(s_n))\setminus X$.
  A substitution is \emph{grounding} for $A$ if it is grounding for all constraint terms $\Gamma_i\parallel s_i$.
\end{definition}

If the terms in a congruence class all have the same constraint then we use $\{ \Gamma\parallel s_1,...,s_n\}$
as a shorthand for $\{\Gamma\parallel s_1,...,\Gamma\parallel s_n\}$. For example, with all shorthands 
we can now write $\{g(x),h(x)\parallel g(x),h(x)\}$ instead of 
$\{g(x)\in \mathcal{M},h(x)\in \mathcal{M}\parallel g(x),g(x)\in \mathcal{M},h(x)\in \mathcal{M}\parallel h(x)\}$. In the 
calculus later on the constraints of each term within a class are always the same. Variables in a class
can always be renamed.

\begin{definition}
  Let $A=\{\Gamma_1\parallel s_1,...,\Gamma_n\parallel s_n\}$ be a class and $\mu$ a substitution. We define $A\mu$ as 
  $\{\Gamma_1\mu\parallel s_1\mu,...,\Gamma_n\mu\parallel s_n\mu\}$. In particular, if $\mu$ is grounding for $A$ we overload its application by
  $A\mu = \{s\mu~|~(\Gamma\parallel s\in A\mathit{\ and\ } \Gamma\mu\mathit{\ true})\}$.
\end{definition}

The semantics of a congruence class with variables is defined
by creating a mapping to the corresponding ground classes. It is important to distinguish between 
separating and free variables here. Separating variables divide the non-ground class into several ground classes.

\begin{definition}[Congruence Class Semantics] \label{def:classgrounding}
Let $A$ be a congruence class. Let $X$ be the separating variables of $A$.
Then the set $\gnd'(A)$ is defined as:
  $$\bigcup_{\substack{\sigma\;\mathit{ground},\\  \dom(\sigma) = X}}~  \bigl\{\{(\Gamma\sigma \parallel s\sigma)~|~ (\Gamma\parallel s) \in A\;\mathit{and}\;\Gamma\sigma\;\mathit{satisfiable}\}\bigr\}$$
and the set $\gnd(A)$ is defined as: 
  $$\bigcup_{B\in \gnd'(A)} \Bigl\{\bigcup_{\sigma\;\mathit{grounding\ for}\; B} \{ s\sigma~|~ (\Gamma\parallel s) \in B\;\mathit{and}\; \Gamma\sigma\;\mathit{true}\}\Bigr\}$$
\end{definition}

\begin{example}
    Assume $\opers =\{g,h,a ,b\}$, $\mathcal{M}=\{g(a),g(b),h(a), h(b)\}$ and
    classes 
    \[A=\{g(x), h(x)\parallel g(x),h(x)\}, B=\{g(x), h(y)\parallel g(x), h(y)\}\]
    Then 
    \[\gnd(A)=\{\{g(a), h(a)\}, \{g(b), h(b)\}\} \text{ and } \gnd(B)= \{\{g(a), h(a), g(b), h(b)\}\}\]
  \end{example}

\begin{definition}[Normal Class]
  Let $A$ be a class. The \emph{normal}
  class $\text{norm}(A)$ is defined as 
  $\{(\Gamma\land\Gamma\sigma\parallel s)~|~(\Gamma\parallel s)\in A\}
  \cup \{(\Gamma\land\Gamma\sigma\parallel s\sigma)~|~(\Gamma\parallel s)\in A$ and $s$ contains free variables$\}$
  for a renaming $\sigma$ on the free variables 
  introducing only fresh variables.
\end{definition}

A class can be turned into a normal class by generating exactly one renamed copy for all constrained terms containing free variables. The motivation for normal classes
is of technical nature. \CCVAR\ rules always operate on two terms out of a class. In case of terms with variables the two terms may be actually instances
of the same term from the class. This can only happen for terms with free variables. By introducing one renamed copy for such terms, the style of \CCVAR\ rules
is preserved and the rules do not need to distinguish between free and separated variables.
For example, the class $A = \{g(x),h(y)\parallel g(x), h(y)\}$ contains all ground terms build with top-symbols $g$ and $h$. So
for a \CCVAR\ step picking $g(a)$ and $g(b)$ out of the class the term $g(x)$ needs to be instantiated with two different constants.
By using renamed copies
$\text{norm}(A) = \{g(x),g(x'),h(y),h(y')\parallel g(x), g(x'), h(y), h(y')\}$ this technical issue is removed. Obviously, $\gnd(A) = \gnd(\text{norm}(A))$, holding
for all classes and their normal counterparts.



\begin{definition}[Subsumption] \label{def:subsumption}
  A class $B$ subsumes another class $A$ if for all $A'\in \gnd(A)$ there exists a $B'\in \gnd(B)$ such that $A'\subseteq B'$.
\end{definition}




The following definitions and lemmas prepare termination, Lemma~\ref{lem:cccterm}.
We show that we cannot create infinitely many classes restricted by $\mathcal{M}$ such that each new class is not subsumed 
by any existing class, guaranteeing termination.

\begin{definition}
  Let $A=\{\Gamma_1\parallel s_1,...,\Gamma_n\parallel s_n\}$ be a class. 
  $A$ is constrained by $\mathcal{M}$ (or $\mathcal{M}$-constrained)
  iff $(s_i{\in}\mathcal{M})\in \Gamma_i$ for all $1\leq i\leq n$.
\end{definition}

\begin{restatable}[Constraint Semantics matches $\mathcal{M}$]{lemma}{nongndtognd}
  \label{lem:nongndtognd}
  Let $A$ be a $\mathcal{M}$-constrained class.
  Then $\gnd(A)\in \mathcal{P}(\mathcal{M})$, the powerset of $\mathcal{M}$.
\end{restatable}

\begin{restatable}[Constraint Semantics matches $\mathcal{M}$]{lemma}{finnumclscomb}
  \label{lem:finnumclscomb}
  There exists no infinite chain of (possibly non-ground) $\mathcal{M}$-constrained classes $A_0,A_1,...$
  such that for all $i\geq 0$, $A_i$ is not subsumed by any $B\in\{A_0,...,A_{i-1}\}$.
\end{restatable}


\begin{definition}
Let $A=\{\Gamma_1\parallel s_1,...,\Gamma_n\parallel s_n\}$ be a class. The set $vars(A)$ is defined as
$\bigcup_{1\leq i\leq n} (\bigcup_{t\in\Gamma_i} vars(t)) \cup vars(s_i)$. 
\end{definition}

A \emph{state} of \CCVAR~is a finite set of congruence classes.
For a state $\Pi = \{A_1,...,A_n\}$ we define $\gnd(\Pi) = \gnd(A_1) \cup...\cup \gnd(A_n)$.
For any state $\Pi$ and classes $\{A,B\}\subseteq\Pi$ with $A\not=B$, we assume that $vars(A)\cap vars(B)=\emptyset$.
This can always be achieved by renaming the variables of the classes.
Now given a set of equations $E$, where $\gnd_{\mathcal{M}}(s\approx t)\not=\emptyset$ 
for all $s\approx t\in E$ the initial state of \CCVAR~is 
\begin{equation*}
  \begin{split}
\Pi = &\{ \{s\in\mathcal{M}\land t\in\mathcal{M}\parallel s,s\in\mathcal{M}\land t\in\mathcal{M}\parallel t\}~|~s\approx t\in E\}~\cup~ \\
&\{\{f_i(x_{1_i},\ldots,x_{k_i})\in\mathcal{M}\parallel f_i(x_{1_i},\ldots,x_{k_i})\}~|~f_i\in\opers\}
\end{split}
\end{equation*}

In particular, the linear single term classes $f_i(x_{1_i},\ldots,x_{k_i})$ are needed for technical reasons. They enable the Deduction rule to build
terms that are not contained in $E$ as a subterm but are contained in $\mathcal{M}$.

\begin{definition}[$\Pi$-\emph{closed}]
  We call a set $\mathcal{M}$ of ground terms $\Pi$-\emph{closed} for a ground term $t$ and a set of congruence classes $\Pi$, if 
  for all (sub)terms $s$ occurring in $\Pi$ and all grounding substitutions $\sigma$ for $s$  we have if $s\sigma\in \mathcal{M}$ then
  $s\delta\in \mathcal{M}$ for all $\delta$ where if $x\sigma \neq x\delta$ for some variable $x$, then $x\delta = t$.
\end{definition}

The ground term set $\mathcal{M}$ from the introduction is  $\Pi$-\emph{closed}  with respect to
$\{g(x)\parallel g(x)\}, \{h(x)\parallel h(x)\}, \{f(x)\parallel f(x)\},  \{g(x), h(x)\parallel g(x), h(x)\}, \{h(h(y))\allowbreak, f(h(y))\parallel\allowbreak h(h(y)),\allowbreak f(h(y))\}$
for any of the constants. We need this property for our soundness proof, in particular Lemma~\ref{lem:groundsound}. The property guarantees that the
rules do not generate classes that have no meaning in $\mathcal{M}$. From now on we assume that $\mathcal{M}$ is always $\Pi$-\emph{closed}
for some ground term $t\in M$.
We present our algorithm in the form of two abstract rewrite rules:

\vspace{4pt}

\shortrules{Merge}
{$\Pi\cup \{A,B\}$}
{$\Pi\cup \{A,B,(A'\cup B')\mu\}$}
{provided 
 $\text{norm}(A)=\{\Gamma_1\parallel s_1,...,\Gamma_n\parallel s_n\}$,
 $\text{norm}(B)=\{\Delta_1\parallel t_1,...,\Delta_n\parallel t_n\}$,
 there
 exist $(\Gamma\parallel s)\in A,(\Delta\parallel t)\in B$ and $\mu$ such that $\mu = mgu(s,t)$,
 $A'=\{\Gamma_1\land \Gamma\land\Delta\parallel s_1,..., \Gamma_n\land \Gamma\land\Delta\parallel s_n\}$,
 $B'=\{\Delta_1\land \Gamma\land\Delta\parallel t_1,..., \Delta_n\land \Gamma\land\Delta\parallel t_n\}$,
 there exists no $A''\in \Pi\cup \{A,B\}$ such that $(A'\cup B')\mu$ is subsumed by $A''$.
 }{\CCVAR}{6}

 \vspace{4pt}

The rule \emph{Merge} takes as input two classes where a term in the first class is unifiable with a term 
in the second class. The result is the union of their normal classes with the unifier applied. For termination
it is crucial to check if there exists a class that subsumes the newly generated class. Note that the 
\emph{Merge} rule can be seen as a generalization of the \emph{Merge} rule in $\Rightarrow_{\text{CC}}$.
We do not explicitly require the constraint of the newly generated class to be satisfiable, if it is not, the class
is empty and subsumed by any other class.

\vspace{4pt}



 \shortrules{Deduction}
 {$\Pi\cup \{A,B\}$}
 {$\Pi\cup \{A,B,\{\Gamma'\parallel f(s'_1,...,s'_n),\Gamma'\parallel f(t'_1,...,t'_n)\}\mu\}$}
 {provided
  $\Gamma\parallel f(s_1,...,s_n)\in A$, $\Delta\parallel f(t_1,...,t_n)\in B$, 
  and for each $0< i\leq n$, 
  there exists a $D_i\in \Pi$ such that $\Gamma_i\parallel s'_i\in \text{norm}(D_i),\Delta_i\parallel t'_i \in \text{norm}(D_i)$,
  $\mu$ is a simultaneous mgu of $f(s'_1,...,s'_n)=f(s_1,...,s_n)$ and $f(t'_1,...,t'_n) = f(t_1,...,t_n)$,
  $\Gamma' = \Gamma\land \Delta\land \Gamma_1\land...\land \Gamma_n \land \Delta_1\land...\land \Delta_n$,
  there exists no $A'\in \Pi\cup \{A,B\}$ such that $\{\Gamma'\parallel f(s'_1,...,s'_n),\Gamma'\parallel f(t'_1,...,t'_n)\}\mu$ is subsumed by $A'$.
  }{\CCVAR}{11}

\smallskip

 The rule \emph{Deduction} creates a new class if there exist terms with the same top symbol in two different (copies of)
 classes such that their arguments are unifiable with terms that are in the same class. Again \emph{Deduction}
 is a generalization of the \emph{Deduction} rule from $\Rightarrow_{\text{CC}}$. Note, that in \emph{Merge} and \emph{Deduction} $A$ and $B$ can be identical.
 In this case we assume the consideration of a renamed copy. Furthermore, in both rules we inherit all parent constraints
 for the new class. Using this invariant, we could
 also represent each class by a single constraint that is not dedicated to a term. We do not do so in order to get a nicer representation,
 the way constraints are composed depending on the term they belong to.
 Our way of constraint composition can  also result in constraints containing
 variables that do not occur in any term of the class anymore. We'll take care of these constraints in Section~\ref{sec:implementation}.


 

 Note that a large number of redundant classes can be created. 
 While these redundant classes affect neither soundness nor completeness getting rid
 of redundant classes is essential for an efficient implementation. To this end we introduce
 a Subsumption rule that has precedence 
 over all other rules.

\smallskip

 \shortrules{Subsumption}
 {$\Pi\cup \{A,B\}$}
 {$\Pi\cup \{B\}$}
 {provided $B$ subsumes $A$.}{\CCVAR}{14}

We will now prove termination, soundness and completeness. We start with soundness.

\begin{restatable}[Normal Classes]{lemma}{normaltoground}
  \label{lem:normaltoground}
  Let $A$ be a class. There exists a $A'\in gnd(A)$ s.t.
  for all $s, t$: $\{s,t\}\subseteq A'$ iff there exists a substitution $\sigma'$ such that $\{s,t\}\subseteq \text{norm}(A)\sigma'$.
\end{restatable}


\begin{restatable}[Ground Testing]{lemma}{groundsound}
  \label{lem:groundsound}
Let $\Pi$ be a set of classes. Assume that $\mathcal{M}$ 
is $\Pi$-closed for the term $\alpha$. Let $A\in\Pi$ be a class where all constraints are the same, 
$X$ separating variables, $Y$ free variables of $A$
and $gnd_{\mathcal{M}}(E)\models s\sigma\approx t\sigma$ for all $\{s\sigma,t\sigma\}\subseteq A\sigma$ and grounding $\sigma$.
Then for all $\Gamma\parallel s\in A$ we have $gnd_{\mathcal{M}}(E)\models s\sigma\delta\approx s\sigma\delta'$ for all 
$\sigma:X\rightarrow \mathcal{T}(\Omega, \emptyset)$, $\delta:Y\rightarrow \mathcal{T}(\Omega, \emptyset)$ and 
$\delta' = \{y\mapsto \alpha~|~ y\in Y\}$, where $\alpha\in \mathcal{M}$, $\Gamma\sigma\delta$ satisfiable.
\end{restatable}

\begin{restatable}[]{corollary}{nonnormtonorm}
  \label{lem:nonnormtonorm}
  Let $A$ be a class where all constraints are the same, 
and $gnd_{\mathcal{M}}(E)\models s\sigma\approx t\sigma$ for all $\{s\sigma,t\sigma\}\subseteq A\sigma$ and grounding $\sigma$.
Then $gnd_{\mathcal{M}}(E)\models s\sigma\approx t\sigma$ for all $\{s\sigma,t\sigma\}\subseteq \text{norm}(A)\sigma$ and grounding $\sigma$.
\end{restatable}

\begin{restatable}[$\Rightarrow_{\CCVAR}$ is sound]{lemma}{cccsound}
 \label{lem:cccsound}
  For any run of $\Rightarrow_{\CCVAR}$, any state $\Pi'$ in this run and
  for all terms $s,t$ and grounding substitution $\sigma$ such that there exists a class 
  $A\in \Pi'$ such that $\{\Gamma\parallel s,\Delta\parallel t\}\subseteq \text{norm}(A)$, $\Gamma\sigma$ and $\Delta\sigma$ satisfiable,
  it holds that 
  $\gnd_{\mathcal{M}}(E)\models s\sigma\approx t\sigma$ .
\end{restatable}

\begin{restatable}[$\Rightarrow_{\CCVAR}$ is Terminating]{lemma}{cccterm}
 \label{lem:cccterm}
  For any run of $\Rightarrow_{\CCVAR}$ we reach a state, where no rule of $\Rightarrow_{\CCVAR}$ 
  is applicable anymore.
\end{restatable}

\begin{restatable}[$\Rightarrow_{\CCVAR}$ is Complete]{lemma}{ccccomp}
 \label{lem:ccccomp}
  Let $\Pi$ be the result of a run of $\Rightarrow_{\CCVAR}$ such that no rule of 
  $\Rightarrow_{\CCVAR}$ is applicable anymore. Then for all 
  $\{s,t\}\subseteq \mathcal{M}$  such that $\gnd_{\mathcal{M}}(E)\models s\approx t$
  there exists a class $A\in\Pi$, $\{\Gamma\parallel s',\Delta\parallel t'\}\subseteq \text{norm}(A)$ and grounding substitution $\sigma$ 
  such that $s'\sigma = s$ and $t'\sigma = t$ and $\Gamma\sigma, \Delta\sigma$ satisfiable.
\end{restatable}

\section{Towards Implementation}
\label{sec:implementation}

In the following, we describe refinements and optimizations of the \CCVAR\ calculus towards an implementation. We start with a simplification of 
the Deduction rule.
It suffices to use only the single term classes for $A$ and $B$ in the Deduction rule.

\begin{restatable}[Deduction with Single Term Classes]{lemma}{dedstc}
  Assume a state $\Pi$ such that $Deduction$ is applicable for some $\{A,B\}\subseteq \Pi$ and $\Gamma\parallel f(s_1,...,s_n)\in A, 
  \Delta\parallel f(t_1,...,t_n)\in B$. Let $\mu$ be the resulting simultaneous mgu and 
  $\{\Gamma_i \parallel s'_i,\Delta_i\parallel t'_i\}\subseteq \text{norm}(D_i)$ for all $1\leq i\leq n$ such that 
  $\{\Gamma'\parallel f(s'_1,...,s'_n), f(t'_1,...,t'_n)\}\mu$ is the resulting new class.
  Then Deduction is applicable for two variable disjoint copies of the single term class $\{f(x_1,...,x_n)\parallel f(x_1,...,x_n)\}$.
\end{restatable}

It is easy to see that the class created by the single term class is always more general than the other classes that can be created by 
Deduction. 

A naive implementation of Subsumption, Definition~\ref{def:subsumption}, by ground instantiation results in
a practically intractable procedure. Therefore, it is approximated by the below subsumption by matching rule
that does not need ground instantiation and is practically tractable, see Section~\ref{sec:evaluation}.

\begin{definition}[Subsumption by Matching]\label{def:submatch}
  Let $A, B$ be classes. Let $X$ be the separating variables of $B$ and $Y$ the free variables 
of $B$.
  $B$ subsumes $A$ by matching iff there exists a substitution $\sigma:X\rightarrow 
\mathcal{T}(\Omega, vars(A))$ that maps every variable in $X$, such that for every 
$\Gamma\parallel t\in A$ there is a $\tau:Y\rightarrow 
\mathcal{T}(\Omega, vars(A))$ and $(\Delta\parallel s)\sigma\in B\sigma$ such that $t=s\sigma\tau$
and $\forall\delta. (\Gamma\delta\rightarrow\exists \delta'.\Delta\sigma\tau\delta\delta')$.
\end{definition}

\begin{restatable}{lemma}{subsumptionmatch}
  \label{lem:subsumptionmatch}
  Let $A, B$ be classes. If $B$ subsumes $A$ by matching then $B$ subsumes $A$.
\end{restatable}

The converse of Lemma~\ref{lem:subsumptionmatch} does not hold.
Consider symbols 
$f,g,a,b$, $\mathcal{M}=\{a,\allowbreak b,\allowbreak f(a),\allowbreak f(b),\allowbreak g(a),\allowbreak g(b)\}$ and classes 
$A=\{f(x),g(a),g(b)\parallel\allowbreak  f(x),g(a), g(b)\}$ and 
$B=\{f(a),f(b),g(x)\parallel\allowbreak  f(a),f(b), g(x)\}$. Then neither $A$ subsumes $B$ nor $B$
subsumes $A$ by subsumption by matching, although $gnd(A)=gnd(B)$.
However, in Section~\ref{sec:evaluation} we show that the subsumption by matching rule
performs nicely in practice.

Before creating a new class we rename all involved classes and then apply \emph{Merge} or \emph{Deduction}.
Especially, after application of \emph{Deduction} the new class
may contain variables in a constraint that do not occur in any of the class terms.
Subsequent merges then continuously increases the number of constraints and variables. Fortunately,
only one extra variable is needed.

\begin{restatable}{lemma}{reduceconstraints}
\label{lem:reduceconstraints}
Let $A=\{\Gamma\parallel s_1,..., s_n\}$ be a class. Let 
$Y = vars(\Gamma)\setminus (vars(s_1)\cup...\cup vars(s_n))$ be the variables
occurring in $\Gamma$ but not in $s_1,...,s_n$. Let $\sigma:Y\rightarrow \{y'\}$ for
some fresh variable $y'\in\mathcal{X}$.
Then $A$ subsumes $A\sigma$ and $A\sigma$ subsumes $A$.
\end{restatable}

To keep the number of constrained terms within classes small, we also need a new condensation rule:
\smallskip

\shortrules{Condensation}
{$\Pi\cup \{\{\Gamma_1\parallel s_1,...,\Gamma_n\parallel s_n\}\}$}
{$\Pi\cup \{\{\Gamma_1\parallel s_1,...,\Gamma_{j-1}\parallel s_{j-1},\Gamma_{j+1}\parallel s_{j+1},...,\Gamma_n\parallel s_n\}\delta\}$}
{provided there exists indices $i,j$ and a matcher $\delta$ such that $s_i\delta = s_j$,
 $\{\Gamma_1\parallel s_1,...,\Gamma_{j-1}\parallel s_{j-1},\Gamma_{j+1}\parallel s_{j+1},...,\Gamma_n\parallel s_n\}\delta$
 subsumes $\{\Gamma_1\parallel s_1,...,\Gamma_n\parallel s_n\}$.
 }{\CCVAR}{14}

\smallskip

For example, the Condensation rule would reduce the class 
$\{f(x),\allowbreak f(y),\allowbreak f(z)\parallel\allowbreak  f(x),\allowbreak  f(y),\allowbreak  f(z)\}$ to $\{f(x),f(y)\parallel f(x), f(y)\}$. Condensation 
together with subsumption by matching ensures termination, because the number of separating variables is bounded.
Condensation is an example where we could improve the performance of our
implementation. It
can be implemented without additional memory consumption, however, our current implementation copies the
class, modifies it and then checks subsumption.

To keep the number of full subsumption checks to a minimum, we have also implemented fast pre-filtering techniques.
It turns out that it is more efficient 
to first check whether for each term in the instance class there is a term in the general class 
that matches this term.

\begin{corollary}
  \label{coro:subfilgenterm}
  Let $A,B$ be two classes. If there exists a $\Gamma\parallel t\in A$ such that there exists no
  $\Delta\parallel s\in B$ and no matcher $\delta$ such that $s\delta = t$, then $A$ is not 
  subsumed by $B$.
\end{corollary}


We use bit vectors to track the number of occurrences of top symbols of terms within a class
to check the above filter.
For each symbol we store in one bit whether there are 0, 1 or more terms in the class that contain 
this symbol as a top symbol, where $[0]_{10} = [0]_2, [1]_{10} = [1]_2$.
For two bit vectors $\mathcal{V}_0$ of a general and $\mathcal{V}_1$ of an instance class 
we compute NOT$(\mathcal{V}_0)$ AND $\mathcal{V}_1$, where NOT and AND are 
bitwise operators. Note that classes containing a variable term must be excluded from this check.

\section{Evaluation}
\label{sec:evaluation}

The first step for an actual implementation and evaluation is to choose a representation for the
set of ground terms $\mathcal{M}$. This could be through an explicit representation
of the ground terms, e.g., by indexing~\cite{NieuwenhuisHRV01} or an implicit representation, e.g.,
by an upper bound with respect to an ordering. For our evaluation we choose the latter, in particular
we define the set $\mathcal{M}$ with respect to the ordering $\preceq$ that simply counts the number
of symbols. Given a finite signature and a finite set of variables and a ground term $\beta$, there are only finitely many different
ground terms smaller than $\beta$, denoted by $\term_{\preceq\beta}$ and  we define $\mathcal{M} = \term_{\preceq\beta}$.
More formally, let $s,t$ be two terms. Then $s\preceq t$ if $\size(s)\leq \size(t)$.
Now $t\in \mathcal{M}$ iff $t\in \term_{\preceq\beta}$ iff $t\preceq\beta$.
The ground term set $\mathcal{M}$ defined this way is $\Pi$-closed with respect to any set of congruence classes
and the constants in $M$, because the constants are minimal in the symbol order.

Solving $\preceq$-constraints is equivalent to
solving linear integer arithmetic constraints.
Note that a constraint $\Gamma$ is satisfiable iff it is satisfiable without applying any
substitution, because variables and constants are the smallest terms in the symbol count order.
Regarding the implication 
$\forall\delta. (\Gamma\delta\rightarrow\exists \delta'.\Delta\sigma\tau\delta\delta')$, see Definition~\ref{def:submatch}, the quantifier alternation can be removed, so 
$\forall\delta. (\Gamma\delta\rightarrow\exists \delta'.\Delta\sigma\tau\delta\delta')$ holds
iff $\forall\delta. (\Gamma\delta\rightarrow\Delta\sigma\tau\delta)$ holds. This is due 
to the fact that there exists such a $\delta'$ in the symbol count order iff all variables are 
replaced by a smallest term. But again, according to the definition of 
the symbol count order, the variable is already a smallest term.

\begin{definition}[LIA Constraint Abstraction]
  \label{ccx:liaconstr}
  Let $t\preceq\beta$ be a constraint. Let $vars(t) = \{x_1,...,x_n\}$. Then 
  $\lic(t\preceq\beta)= x_1\geq 1\land ...\land x_n\geq 1\land \#(x_1,t)*x_1 + ... + 
  \#(x_n,t)*x_n\leq \size(\beta) - (\size(t)- \sum_{1\leq i\leq n}\#(x_i,t))$
  is the linear arithmetic constraint of $t\preceq\beta$.
\end{definition}

\begin{restatable}[Correctness LIA Constraint Abstraction]{lemma}{corcoabstr}
  \label{lem:corcoabstr}
  Let $t\preceq\beta$ be a constraint, $vars(t) = \{x_1,...,x_n\}$.
  \begin{enumerate}
    \item For any ground substitution $\sigma = \{x_i\mapsto s_i ~|~ 1\leq i \leq n\}$:
      if $t\sigma\preceq\beta$ is  true then $\lic(t\preceq\beta)\{x_i\mapsto \size(s_i)~|~1\leq i\leq n\}$ is true.
    \item For any substitution  $\sigma = \{x_i\mapsto k_i ~|~ 1\leq i\leq n, k_i\in\Nat\}$:
      if $\lic(t\preceq\beta)\sigma$ is true then $t\delta\preceq\beta$ is true for
      all $\delta = \{x_i\mapsto s_i ~|~ 1\leq i\leq n\}$ where all $s_i$ are ground, and $\size(s_i) = k_i$.
      \end{enumerate}
    \end{restatable}

    \begin{restatable}[LIA Quantifier Elimination]{lemma}{cotcoelim}
      \label{lem:cotcoelim}
  Let $\beta$ be a ground term and 
  $\Gamma_1 = \{s_1\preceq \beta\land...\land s_n\preceq \beta\}$ and
  $\Gamma_2 = \{t_1\preceq \beta\land...\land t_m\preceq \beta\}$ be two 
  constraints. Then $\forall\sigma. (\Gamma_1\sigma\rightarrow\exists 
  \sigma'.\Gamma_2\sigma\sigma')$ iff 
  $\lic(s_1\preceq \beta)\land ...\land \lic(s_n\preceq \beta)\rightarrow 
  \lic(t_1\preceq \beta)\land ...\land \lic(t_m\preceq \beta)$.
\end{restatable}

Thus checking if a $\preceq$-constraint $\Gamma$ models a $\preceq$-constraint $\Gamma'$ 
reduces to a linear integer arithmetic implication test. We make use of the linear arithmetic solver
implemented in \textsc{Spass}-SATT~\cite{BrombergerEtAl19}. 

With regard to reducing the number of constraints as seen in Lemma \ref{lem:reduceconstraints} we can further 
reduce the number of constraints in our specific case of $\preceq$ by removing terms that have the same variable occurrences
and are smaller or equal to some other term in the constraint according to $\preceq$,
e.g. $f(x)$ can be removed if $f(g(x))$ already exists.

We process classes using a \emph{worked-off} and \emph{usable} queue.
Initially the \emph{worked-off} queue contains all single-term classes and \emph{usable} contains the initial classes
that are created from the input equations. In each loop we select a class from the \emph{usable} queue, 
perform all possible \emph{Merge} and \emph{Deduction} steps on it and add the class to the 
\emph{worked-off} queue afterwards if not subsumed. Newly created classes are added to the usable queue.

Finally, we record whether merge or subsumption checks have already been applied on classes. To find 
candidates for the subsumption rule we maintain an index for each term in which classes it occurs.
General terms are retrieved by a discrimination tree index, unifiable terms and instances of a term
by a path index~\cite{McCune92a}.

We evaluated our algorithm on all unit equality (UEQ) problems from TPTP-v8.2.0~\cite{Sutcliffe17}. 
From each problem we created two benchmark problems: one with all inequations removed and
the other by turning inequations into equations.
We choose two fixed nesting depths of $6$ and $8$ in order to construct $\beta$.
It is constructed by nesting all
function symbols in one argument up to the chosen depth and filling all other 
arguments with a constant.
E.g., with function symbols $a/0$,$f/1$,$g/2$,$h/1$,$i/2$ and nesting depth $4$ we 
create the term $\beta = i(h(g(f(a), a)),a)$. Then the size of all terms is limited to $7$~symbols.

We created another evaluation on the SMT-LIB UF non-incremental problems (release 2024)~\cite{BarFT-SMTLIB} for a nesting
depth of 2 without those examples that use the distinct operator.
In order to make these examples available for untyped  congruence closure, we carried out various pre-processing steps:
we expanded let operators, removed the typing and quantifiers, coded predicates as equations and
did a CNF transformation. We then selected the first equational literal in each clause that contained at least one variable,
excluding equations of the form $x\approx y$, $x\approx t$, or $t\approx x$, as input for \CCVAR. 
If no such literal was found, we instead chose the first ground literal and skipped the clause if this literal was also absent.
We skipped the above mentioned equations because there existence typically resulted in trivial problems in the case of \CCVAR, 
partly because the typing was removed.

We compared the performance of our algorithm to the performance of a \CC\ 
implementation based on the implementation in the veriT solver~\cite{bouton2009verit}.
The resulting benchmark problems turn out to be challenging for both algorithms.

\CCVAR~ provides a 
solution to the entire ground input space smaller than $\beta$.
 Therefore,  for the comparison of the two algorithms, we feed all ground terms smaller 
than $\beta$ into \CC. \CCVAR~also provides a 
solution to the entire ground input space. We skip all examples where no equation has ground 
instances $\preceq\beta$. 

Experiments for the UEQ problems are performed on a Debian Linux server running AMD EPYC 7702 64-core CPUs 
with 3.35GHz and a total memory of 2TB and for the SMT examples they were performed on
AMD EPYC 9754 128-Core Processors with 3.1 Ghz and a total memory of 2.2TB.
The time limit for each test is 30 minutes.
The results of all runs as well as all input files 
and binaries can be found at \url{https://nextcloud.mpi-klsb.mpg.de/index.php/s/ttjj3tDBebgtmHj}
for the UEQ problems and at \url{https://nextcloud.mpi-klsb.mpg.de/index.php/s/769Ca7Dc9Ck88pi} for the UF problems.

For a nesting depth of $6$ on the UEQ problems, \CCVAR~ terminates on 519 
and CC on 457 of the 2900 problems. \CCVAR~is faster on 474 problems and 
\CC\  on 172. \CCVAR\ terminated on 189 examples where \CC\  timed out, and \CC\  terminated on
127 examples where \CCVAR\ timed out.





Figure~\ref{ccx:eval:time6} shows the results of all terminating test cases for the runtime and Figure~\ref{ccx:eval:classes6} shows the results for the number of classes.
The performance of \CCVAR\  currently drops if there are many different variables.
This is mainly due to the current implementation of the 
redundancy checks. Concerning the number of classes, the number of classes generated by
\CCVAR\ is significantly smaller than the number in \CC\  for almost all examples. Examples where this does not hold are border cases, i.e., they only contain few equations or contain only one constant.

\begin{figure}[h]
  \begin{subfigure}[h]{0.45\linewidth}
    \resizebox {\textwidth} {!} {
    \begin{tikzpicture}
      \begin{axis}[xmode=log,ymode=log, xlabel={Time \CCVAR},
                   y label style={at={(axis description cs:0.1,.5)},anchor=south},
                   x label style={at={(axis description cs:0.5,0.02)},anchor=north},
                  ylabel={Time \CC},
                  unit vector ratio*=1 1 1, 
                  samples=10000,
                  xmin=0.1,
                  ymin=0.1,
                  ymax=5000,
                  xmax=5000]
      \addplot+ [color=black,mark=x] table [only marks, col sep=comma,x index=1, y index=2] {comparisontime6.txt};  
      \addplot [color=black,no marks] coordinates {(0.000001, 0.000001) (5000, 5000)};
      \end{axis}
    \end{tikzpicture}
    }
  \caption{\label{ccx:eval:time6}Comparison of the runtime of \CC\  and \CCVAR.}
  \end{subfigure}
  \hfill
  \begin{subfigure}[h]{0.45\linewidth}
    \resizebox {\textwidth} {!} {
    \begin{tikzpicture}
      \begin{axis}[xmode=log,ymode=log, xlabel={\#Classes \CCVAR}, ylabel={\#Classes \CC\ },
        y label style={at={(axis description cs:0.1,.5)},anchor=south},
        x label style={at={(axis description cs:0.5,0.02)},anchor=north},
                   unit vector ratio*=1 1 1, samples=10000,xmin=1,ymin=1,ymax=3000000,xmax=3000000]
      \addplot+ [color=black,mark=x] table [only marks, col sep=comma,x index=2, y index=1] {comparisonclasses6.txt};  
      \addplot [color=black,no marks] coordinates {(0.000001, 0.000001) (3000000, 3000000)};
      \end{axis}
    \end{tikzpicture}
    }
  \caption{\label{ccx:eval:classes6}Comparison of the number of classes of \CC\  and \CCVAR. }
  \end{subfigure}%
  \caption{Benchmark results for a nesting depth of 6.}
\end{figure}

\begin{figure}[h]
  \begin{subfigure}[h]{0.45\linewidth}
    \resizebox {\textwidth} {!} {
      \begin{tikzpicture}
        \begin{axis}[xmode=log,ymode=log, xlabel={Time \CCVAR}, ylabel={Time \CC},
                     y label style={at={(axis description cs:0.1,.5)},anchor=south},
                     x label style={at={(axis description cs:0.5,0.02)},anchor=north},
                     unit vector ratio*=1 1 1, samples=10000,xmin=0.1,ymin=0.1,ymax=5000,xmax=5000]
        \addplot+ [color=black,mark=x] table [only marks, col sep=comma,x index=1, y index=2] {comparisontime8.txt};  
        \addplot [color=black,no marks] coordinates {(0.000001, 0.000001) (5000, 5000)};
        \end{axis}
      \end{tikzpicture}
    }
  \caption{\label{ccx:eval:time8}Comparison of the runtime of \CC\  and \CCVAR.}
  \end{subfigure}
  \hfill
  \begin{subfigure}[h]{0.45\linewidth}
    \resizebox {\textwidth} {!} {
      \begin{tikzpicture}
        \begin{axis}[xmode=log,ymode=log, xlabel={\#Classes \CCVAR}, ylabel={\#Classes \CC},
          y label style={at={(axis description cs:0.1,.5)},anchor=south},
          x label style={at={(axis description cs:0.5,0.02)},anchor=north},        
                     unit vector ratio*=1 1 1, samples=10000,xmin=1,ymin=1,ymax=1000000,xmax=1000000]
        \addplot+ [color=black,mark=x] table [only marks, col sep=comma,x index=2, y index=1] {comparisonclasses8.txt};  
        \addplot [color=black,no marks] coordinates {(0.000001, 0.000001) (3000000, 3000000)};
        \end{axis}
      \end{tikzpicture}
    }
  \caption{\label{ccx:eval:classes8}Comparison of the number of classes of \CC\  and \CCVAR. }
  \end{subfigure}%
  \caption{Benchmark results for a nesting depth of 8.}
  \end{figure}

For a nesting depth of $8$ on the UEQ problems, 299 of 2900 problems terminate in \CCVAR\ 
and 102 in \CC. \CCVAR~is faster in 294 terminating examples and 
\CC\ in 24. \CCVAR\ terminated on 216 examples where \CC\ timed out and \CC\ terminated on
19 examples where \CCVAR\ timed out.





Figure~\ref{ccx:eval:time8} shows the results of all terminating test cases for the runtime and Figure~\ref{ccx:eval:classes8} shows the results for the number of classes.
\CCVAR~is particularly advantageous with a large $\beta$,
where the number of ground equations become very large. The number of classes are again significantly 
smaller than in \CC, except for the already mentioned border case examples.

  For a nesting depth of $2$ on the SMT examples, 1123 of 7516 problems terminate in \CCVAR\ 
  and only 561 in \CC. \CCVAR~is faster in 886 terminating examples and 
  \CC\ in 277. \CCVAR\ terminated on 839 examples where \CC\ timed out and \CC\ terminated on
  40 examples where \CCVAR\ timed out.
  The significantly better performance of \CCVAR\ is mainly due to the fact that the problems 
  in the SMT-LIB contain considerably more symbols, which creates a huge number of ground
  equations even for small nesting depth sizes.
  
  
  
  
  \vspace{2pt}
  
  Figure~\ref{ccx:eval:time2-smt} shows the results of all terminating test cases for the runtime and 
  Figure~\ref{ccx:eval:classes2-smt} shows the results for the number of classes.
  Here, we can see that \CCVAR~is already advantageous with a small $\beta$. The number of classes are again  
  smaller in all cases.
  
  \begin{figure}
    \begin{subfigure}[h]{0.45\linewidth}
      \resizebox {\textwidth} {!} {
        \begin{tikzpicture}
          \begin{axis}[xmode=log,ymode=log, xlabel={Time \CCVAR}, ylabel={Time \CC},
                       y label style={at={(axis description cs:0.1,.5)},anchor=south},
                       x label style={at={(axis description cs:0.5,0.02)},anchor=north},
                       unit vector ratio*=1 1 1, samples=10000,xmin=0.1,ymin=0.1,ymax=5000,xmax=5000]
          \addplot+ [color=black,mark=x] table [only marks, col sep=comma,x index=1, y index=2] {comparisontime2-smt.txt};  
          \addplot [color=black,no marks] coordinates {(0.000001, 0.000001) (5000, 5000)};
          \end{axis}
        \end{tikzpicture}
      }
    \caption{\label{ccx:eval:time2-smt}Comparison of the runtime of \CC\  and \CCVAR.}
    \end{subfigure}
    \hfill
    \begin{subfigure}[h]{0.45\linewidth}
      \resizebox {\textwidth} {!} {
        \begin{tikzpicture}
          \begin{axis}[xmode=log,ymode=log, xlabel={\#Classes \CCVAR}, ylabel={\#Classes \CC},
            y label style={at={(axis description cs:0.1,.5)},anchor=south},
            x label style={at={(axis description cs:0.5,0.02)},anchor=north},        
                       unit vector ratio*=1 1 1, samples=10000,xmin=1,ymin=1,ymax=1000000,xmax=1000000]
          \addplot+ [color=black,mark=x] table [only marks, col sep=comma,x index=2, y index=1] {comparisonclasses2-smt.txt};  
          \addplot [color=black,no marks] coordinates {(0.000001, 0.000001) (3000000, 3000000)};
          \end{axis}
        \end{tikzpicture}
      }
    \caption{\label{ccx:eval:classes2-smt}Comparison of the number of classes of \CC\  and \CCVAR. }
    \end{subfigure}%
    \caption{Benchmark results for a nesting depth of 2 on the SMT UF examples.}
  \end{figure}

The following table shows the average and median time and number of classes of examples terminating in both \CCVAR\ and \CC\ 
for nesting depth $6$, $8$ and $2$ in the SMT case.
For the UEQ examples, one can see that \CCVAR\ is significantly faster on average and produces only a fraction of the classes 
of the ground \CC.
The difference is even stronger when looking at the median. Here \CCVAR\ only needs 1 second or 
less for half of all terminating examples whereas \CC\ needs more than ten minutes.
In the SMT case, the numbers are slightly worse, but \CCVAR\ still only takes half the 
time both on average and on median and produces significantly fewer classes.


\centerline{\renewcommand{\arraystretch}{1.2}
\begin{tabular}{r|r|r|r|r|r|r|r|r}
  \multicolumn{1}{c}{} &  \multicolumn{4}{c}{Average}     &  \multicolumn{4}{c}{Median}  \\
    \multicolumn{1}{c}{} &  \multicolumn{2}{c}{Time}     &  \multicolumn{2}{c}{\#Classes} 
    &  \multicolumn{2}{c}{Time}     &  \multicolumn{2}{c}{\#Classes}   \\
    \multicolumn{1}{c|}{Nesting Depth} &  \multicolumn{1}{c|}{\CCVAR}  & \multicolumn{1}{c|}{\CC}   & \multicolumn{1}{c|}{\CCVAR} & \multicolumn{1}{c|}{\CC}  &  \multicolumn{1}{c|}{\CCVAR}  & \multicolumn{1}{c|}{\CC}    & \multicolumn{1}{c|}{\CCVAR} & \multicolumn{1}{c}{\CC} \\ \hline
  6 &  155  &   352.38  & 1199 &    55436 &  0.06  &   57.2  & 107 &    11819   \\
  8  &  94.14 &  720.52    &   2844  & 145705  &  0.001 &  795.73    &   6  & 165145 \\
  2-SMT  &  83.26 &  128.11    &   7731  & 102024  &  2.27 &  5.68    &   2093  & 14958 \\
\end{tabular}}

\vspace{4pt}

\section{Discussion and Conclusion}
\label{sec:discussion}

We presented the new calculus non-ground congruence closure (\CCVAR). 
It takes as input non-ground equations and computes the corresponding congruence classes
respect to a finite set of ground terms $\mathcal{M}$.
The algorithm is sound, complete, and terminating due to a notion of redundancy
and the implicit, finite ground input space.
We implemented redundancy concepts, e.g.,
by introducing filters to expensive checks such as subsumption.
The resulting implementation of \CCVAR\ outperforms \CC\ in the vast
majority of benchmark examples taken from the TPTP~\cite{Sutcliffe17} and SMT-LIB~\cite{BarFT-SMTLIB}.

Still there is room for further improvement.
From an implementation point of view,
as already mentioned, \emph{Condensation} modifies a copy 
of the class and checks for subsumption. 
In the \emph{Merge} and \emph{Deduction} rules, the number of copies of classes is also higher than necessary, 
especially when the generated class is subsumed.
For some border cases \CC\ outperforms \CCVAR. Extending $\Rightarrow_{\CCVAR}$ to cope
with input equations with only a few constants or few equations is a further line of research.
Already now \CCVAR\ can decide shallow equational classes, where the only arguments to functions
are variables. This is independently of $\beta$ and due to our notion of redundancy.

Equality checking between ground terms amounts to instance finding in a particular class, once
the congruence closure algorithm is finished, both for \CC\ and \CCVAR\ where for the latter
this has to be done modulo matching.
Checking the equality of  non-ground terms is much more involved both for \CCVAR\ and \CC. This is mainly due to the fact that we
consider a finite  signature.
If two non-ground terms are not in the same class this does not actually mean that they are not equal,
since it could be the case that all ground instances of these terms are in the same class.

In general, \CCVAR~ outperforms \CC\  if the number of different variables in equations is smaller
than the number of ground terms to be considered. The other way round, if there are many variables but only a few ground
terms to consider, then running \CC\  is beneficial. This situation can be easily checked in advance,
so \CCVAR~can be selected on problems where \CC\  will fail extending the overall scope of applicability.


%
%
%
%
\bibliographystyle{splncs04}
\bibliography{paper}

\newpage

\section*{Appendix}
\subsection{Further Examples}

\begin{example}
  Suppose we have the following equations: $g(x)\approx a$, $h(y)\approx a$, $g(h(z))\approx h(h(z))$.
  Initially, without single term classes, we get 
  $$\Pi = \{\{g(x), a\parallel g(x),a\}, \{h(y), a\parallel h(y),a\},
  \{g(h(z)), h(h(z))\parallel g(h(z)), h(h(z))\}\}$$
  Merging the last class with the second class we get 
  $$\{g(h(z)),\allowbreak h(h(z)),\allowbreak  h(y),\allowbreak  a\parallel g(h(z)),\allowbreak  h(h(z)),\allowbreak  h(y),\allowbreak  a\}$$
  We can now merge this new class with the first class to get:
  $$\{g(h(z)),\allowbreak  h(h(z)),\allowbreak  a,\allowbreak  h(y),\allowbreak  
  g(x)\parallel g(h(z)),\allowbreak  h(h(z)),\allowbreak  a,\allowbreak  h(y),\allowbreak  g(x)\}$$
  This class subsumes all other classes, for example in our ordering defined in 
  section~\ref{sec:implementation}. So this would be the final result. Note that 
  this result is independent of the chosen $\beta$. No matter how large $\beta$ is
  the result of the calculus is always the same (assuming that $\beta$ allows for the 
  initial classes).
  \end{example}
  
  Another example where the number of classes is dependent on $\beta$ is shown below.
  \begin{example}
    Suppose we have the single equation $f(x)\approx g(x)$. 
    Initially we have $\Pi = \{\{f(x),g(x)\parallel f(x),g(x)\}\}$.
  Depending on the size of $\beta$ we get more and more classes with the deduction
    rule, like $\{f(f(x)),f(g(x)), f(x), g(x)\parallel f(f(x)),f(g(x))\}$ and 
    $\{g(f(x)),g(g(x)), f(x), g(x)\parallel g(f(x)),g(g(x))\}$. Which we can again merge with 
    the first class to get 
    $$\{f(f(x)),f(g(x)), g(f(x)),g(g(x)), f(x), g(x)\parallel$$ 
    $$f(f(x)),f(g(x)),g(f(x)),g(g(x))\}$$
    Increasing $\beta$ further we get even larger classes. This is an example 
    where we gain quite little compared to congruence closure.
  \end{example}
  
  The following example shows that Merge also has to be applied to two instances of the same class.
  \begin{example}
  Assume initial classes (without single term classes):
  
  $$\Pi = \{\{f(x,y), g(x,y)\parallel f(x,y), g(x,y)\},$$ 
  $$\{f(x,y), g(y,x)\parallel f(x,y), g(y,x)\}\}$$
  
  Now we can merge the classes by unifying $f(x,y)$:
  
  $$\{f(x,y), g(x,y), g(y,x)\parallel f(x,y), g(x,y), g(y,x)\}$$
  
  The new class subsumes both initial classes. To get the final result we have to merge this new class with itself
  by unifying $g(x,y)$ and $g(y,x)$ to get:
  
  $$A=\{f(x,y), g(x,y), g(y,x), f(y,x)\parallel$$ 
  $$f(x,y), g(x,y), g(y,x), f(y,x)\}$$
  
  Otherwise, e.g. $\{f(a,b), f(b,a)\}\not\subseteq A'$ for all $A'\in gnd(A)$ for an appropriate set of function symbols in 
  $\Omega$.
\end{example}

\CCVAR~does not always guarantee fewer or equally many Congruence Classes than \CC. Consider the following example

\begin{example}
  Let $f(a)\approx h(a), g(a)\approx h(a),f(b)\approx h(b), g(b)\approx h(b), f(x)\approx g(x), a\approx f(a), b\approx f(b)$ be some 
  input equations. Further assume 
  $$\mathcal{M} = \{f(a),g(a),h(a),f(b),g(b),h(b),a,b\}$$ 
  
  Initially \CCVAR\ contains a class for the terms in each equation. Note, that Subsumption is not applicable in this state.
  Now, multiple Merge operations are possible, but no matter in which sequence they are applied the result is always
  $$\{\{f(a),h(a),g(a),a\parallel f(a),h(a), g(a), a\}, \{f(b),h(b),g(b),b\parallel f(b),h(b), g(b), b\},$$
  $$\{f(x),g(x)\parallel f(x),g(x)\}\}$$
  Subsumption is not applicable to this set of classes. In ground congruence closure, however, we only get two classes:
  $$\{\{f(a),h(a), g(a), a\}, \{f(b),h(b), g(b), b\}\}$$

\end{example}

\subsection{Pseudo Code}
This chapter shows pseudo code of our implementation. It is intended to give an idea of the implementation and not to
present the concrete details of the implementation. For example, the use of path index and discrimination tree is missing here.

Algorithm~\ref{ccx:mainfunc} shows the initial state and main loop of our implementation. 
There is also an first optimization option here, namely which classes should be picked from the usable queue first.
Our heuristic selects the classes with the fewest terms and the most variables from the usable queue, 
or if the number of terms and variables are equal, then the class with the fewest separating variables.
Various benchmarks have shown that this produces the best results.
\begin{algorithm}
    \caption{Main function of the algorithm}\label{ccx:mainfunc}
    \begin{algorithmic}
    \Function{Main}{$E$}
    \For {all $s\approx t\in E$}
      \State {$C_{new} =$ \textbf{new} CLASS}
      \State $C_{new}{\rightarrow}terms = \{s,t\}$
      \State $C_{new}{\rightarrow}cstrs = \{s,t\}$
      \State us $=$ Push(us, $C_{new}$)
    \EndFor
    \For {all $f\in \Omega$, $arity(f)= n$}
      \State {$C_{new} =$ \textbf{new} CLASS}
      \State $C_{new}{\rightarrow}terms = \{f(x_1,...,x_n)\}$
      \State $C_{new}{\rightarrow}cstrs = \{f(x_1,...,x_n)\}$
      \State wo $=$ Push(wo, $C$)
    \EndFor
    \While{$us\not= \emptyset$}
      \State $C$ $=$ Pop(us)
      \If {Merge($C$) \&\& Deduct($C$)}
        \State wo $=$ Push(wo, $C$)
      \EndIf
    \EndWhile
    \EndFunction
    \end{algorithmic}
    \end{algorithm}

  Algorithm~\ref{ccx:mergefunc} shows the implementation of our merge function.
Here the implementation follows the rules, except that subsumption is checked directly after the new class is created.

\begin{algorithm}
    \caption{Merge function}\label{ccx:mergefunc}
    \begin{algorithmic}
      \Function{Merge}{$C_0$}
      \For{$C_1$ in $wo$}
        \For{each $(t_0,t_1)\in \{(t_0,t_1)~|~t_0\in C_0{\rightarrow}terms\land t_1\in C_1{\rightarrow}terms\}$}
          \If {unifiable($t_0$,$t_1$)}
            \State {$\mu$ $=$ $mgu(t_0,t_1)$}
            \State {$C_{new} =$ \textbf{new} CLASS}
            \State {$C_{new}{\rightarrow} terms = (C_0{\rightarrow} terms\cup C_1{\rightarrow} terms)\mu$}
            \State {$C_{new}{\rightarrow} cstrs ~= (C_0{\rightarrow} cstrs\cup C_1{\rightarrow} cstrs)\mu$}
            \If {SAT($C_{new}{\rightarrow} cstrs$)}
              \State $subsumed = \mathit{FALSE}$
              \For{$C$ in $wo\cup us\cup \{C_0\}$}
                \If {CheckSubsumption($C$, $C_{new}$)}
                  \State $subsumed = \mathit{TRUE}$
                  \State break
                \EndIf
              \EndFor
              \If {$subsumed == \mathit{FALSE}$}
                \For{$C$ in $wo\cup us\cup \{C_0\}$}
                  \If {CheckSubsumption($C_{new}$, $C$)}
                    \State $wo = wo\setminus \{C\}$
                    \State $us = us\setminus \{C\}$
                    \If {$C == C_0$}
                      $subsumed = \mathit{TRUE}$
                    \EndIf 
                  \EndIf
                \EndFor
    
                \State $us =us \cup \{C_{new}\}$
                \If {$subsumed == \mathit{TRUE}$}
                  \State \textbf{return} $\mathit{FALSE}$
                \EndIf
              \EndIf
            \EndIf
          \EndIf
        \EndFor
      \EndFor
      \State \textbf{return} $\mathit{TRUE}$
      \EndFunction
      \end{algorithmic}
    \end{algorithm}

    Algorithms~\ref{ccx:deductfunc} and~\ref{ccx:deductfuncinternal} show the implementation of our deduction function.
    Note that the constraints are always verified as soon as possible in the actual implementation. 
    To keep the pseudocode simple, we only check the satisfiability of the constraints at the end.
    
    \begin{algorithm}
        \caption{Deduction function}\label{ccx:deductfunc}
        \begin{algorithmic}
          \Function{Deduct}{$C_0$}
            \State $wo = wo\cup \{C_0\}$
            \For {each $f\in \Omega$ with $arity(f)=n$ and $n> 0$}
              \State $\mathit{t_0} = f(x_1,...,x_n)$
              \State $\mathit{t_1} = f(y_1,...,y_n)$
              \If {DeductIntern($t_0$, $t_1$, $0$, $arity(f),C_0,$ \textbf{new} CLASS, $\mathit{false}$) $== \mathit{FALSE}$}
                \State $wo = wo\setminus \{C_0\}$
                \State \textbf{return} $\mathit{FALSE}$
              \EndIf
            \EndFor
          \EndFunction
          \end{algorithmic}
      \end{algorithm}  
      \begin{algorithm}
        \caption{Internal Deduction function}\label{ccx:deductfuncinternal}
        \begin{algorithmic}
          \Function{DeductIntern}{$t_0,t_1,i,n,C_0, C_{new}, used$}
            \If {$i\not= n$}
              \For {$C\in \mathit{wo}$}
                \State $C_{new}{\rightarrow} cstrs = C_{new}{\rightarrow} cstrs \cup C{\rightarrow} cstrs$
                \For {$\{s_0,s_1\}\subseteq C{\rightarrow} terms$}
                  \State $t_0 = t_0[s_0]_i$
                  \State $t_1 = t_1[s_1]_i$
                  \If {not DeductIntern($t_0,t_1,i+1,n,C_0, C_{new}, (C == C_0~||~used)$)}
                    \State \textbf{return} $\mathit{FALSE}$
                  \EndIf
                \EndFor
                \State $C_{new}{\rightarrow} cstrs = C_{new}{\rightarrow} cstrs \setminus C{\rightarrow} cstrs$
              \EndFor
            \ElsIf {$used$}
              \State $C_{new}{\rightarrow} terms = \{t_0,t_1\}$
              \State $C_{new}{\rightarrow} cstrs = C_{new}{\rightarrow} cstrs \cup \{t_0,t_1\}$
              \If {SAT($C_{new}{\rightarrow} cstrs$)}
                \For{$C$ in $wo\cup us$}
                  \If {CheckSubsumption($C$, $C_{new}$)}
                    \State \textbf{return} $\mathit{TRUE}$
                  \EndIf
                \EndFor
                \State $subsumed = \mathit{FALSE}$
                \For{$C$ in $wo\cup us$}
                  \If {CheckSubsumption($C_{new}$, $C$)}
                    \State $wo = wo\setminus \{C\}$
                    \State $us = us\setminus \{C\}$
                    \If {$C == C_0$}
                      $subsumed = \mathit{TRUE}$
                    \EndIf 
                  \EndIf
                \EndFor
                \State $us = us \cup \{C_{new}\}$
                \If {$subsumed == \mathit{TRUE}$}
                  \State \textbf{return} $\mathit{FALSE}$
                \EndIf
              \EndIf
            \EndIf
            \State \textbf{return} $\mathit{TRUE}$
          \EndFunction
          \end{algorithmic}  
      \end{algorithm}

      Algorithm~\ref{ccx:subsumptionfunc} shows our implementation of subsumption. The function\\ $BuildLAC(constraints)$
      creates a linear arithmetic constraint as defined in~\ref{ccx:liaconstr} and $LAImplicationTest(LAC1, LAC0)$
      is the implementation of the implication test of the linear arithmetic solver.

    \begin{algorithm}
    \caption{Subsumption function}\label{ccx:subsumptionfunc}
    \begin{algorithmic}
        \Function{CheckSubsumption}{$C_0,C_1$}
        \If {$C_0{\rightarrow}sepvars == 0$}
        \State \textbf{return} CheckSubsumptionFreeVars($C_0,C_1,\{\}$) 
        \Else
        \For{each $(t_0,t_1)\in \{(t_0,t_1)~|~t_0\in C_0{\rightarrow}terms\land t_1\in C_1{\rightarrow}terms\}$}
            \If {exists $\sigma$ s.t. $t_0\sigma = t_1$}
            \State $\sigma = \sigma\setminus\{x\rightarrow t~|~ x\rightarrow t\in \sigma$ and $x$ a free variable$\}$
            \If {CheckSubsumptionFreeVars($C_0,C_1,\sigma$)}
                \State \textbf{return} $\mathit{TRUE}$
            \EndIf
            \EndIf
        \EndFor
        \EndIf
        \State \textbf{return} $\mathit{FALSE}$
        \EndFunction

        \Function{CheckSubsumptionFreeVars}{$C_0,C_1,\sigma$}
        \For{each $t_1\in C_1{\rightarrow}terms$}
            \State $result = \mathit{FALSE}$
            \For{each $t_0\in C_0{\rightarrow}terms$}
            \If {exists $\delta$ s.t. $t_0\sigma\delta = t_1$}
                \State $LAC_0 = BuildLAC(\{t\sigma\delta~|~t\in C_0{\rightarrow}cstrs\})$
                \State $LAC_1 = BuildLAC(C_1{\rightarrow}cstrs)$
                \If {LAImplicationTest($LAC_1, LAC_0$)}
                \State $result = \mathit{TRUE}$
                \State $break$
                \EndIf
            \EndIf
            \EndFor
            \If {$not$ $result$}
            \State \textbf{return} $\mathit{FALSE}$
            \EndIf
        \EndFor
        \State \textbf{return} $\mathit{TRUE}$
        \EndFunction
    \end{algorithmic}
    \end{algorithm}

\subsection{Missing Proofs}
\label{sec:proofs}

\nongndtognd*
\begin{proof}
  For any $B\in \gnd'(A)$ there exists a $\sigma$ such that $B=\{(\Gamma\sigma \parallel s\sigma)~|~ (\Gamma\parallel s) \in A\;$ and $ \Gamma\sigma$ satisfiable$\}$.
  Thus for any $\Gamma\sigma \parallel s\sigma\in B$ we have $(s\sigma{\in}\mathcal{M})\in \Gamma\sigma$. 
  Thus $B$ is $\mathcal{M}$-constrained.
  For any $A'\in \gnd(A)$ we have $A'= \bigcup_{\sigma\;\mathit{grounding\ for}\; B} 
  \{ s\sigma~|~ (\Gamma\parallel s) \in B\;\mathit{and}\;\Gamma\sigma\;\mathit{true}\}$ for some $B\in \gnd'(A)$. 
  Thus for any $s\sigma\in A'$ we have $s\sigma{\in}\mathcal{M}$, since $B$ is $\mathcal{M}$-constrained and 
  $s\sigma$ is ground. Thus $A'\subseteq \mathcal{M}$.
  Thus $\gnd(A)\subseteq \mathcal{P}(\mathcal{M})$.
\end{proof}

\finnumclscomb*
\begin{proof}
Assume there exists such a chain. Let $\mathcal{P}(\mathcal{M})$ be the 
powerset of $\mathcal{M}$. Since $\mathcal{M}$ is finite, $\mathcal{P}(\mathcal{M})$ is finite as well. 
There are only $2^{|\mathcal{M}|}$ different subsets in $\mathcal{P}(\mathcal{M})$. For any $A$ in the
chain it holds $\gnd(A)\in \mathcal{P}(\mathcal{M})$ by Lemma~\ref{lem:nongndtognd}.
If there exist indices $i\not=j$ such that $\gnd(A_i) = \gnd(A_j)$ then $A_j$ subsumes $A_i$ and vice versa contradicting the assumption. 
Thus, for any pair of indices $i\not=j$ in the chain it has to hold $A_i\not=A_j$. But there are only finitely many different subsets. Contradiction.
\end{proof}

\normaltoground*
\begin{proof}
  Let $X$ be the separating and $Y$ be the free variables of $A$.
  Let $\tau$ be the renaming that maps the free variables to 
  fresh variables to create the normal class. 

  Assume $\{s,t\}\subseteq A'$ for some $A'\in gnd(A)$.
  There must exist terms $\{\Gamma\parallel s',\Delta\parallel t'\}\subseteq A$ and substitutions 
  $\sigma:X\rightarrow \mathcal{T}(\Omega,\emptyset), 
  \delta:Y\rightarrow \mathcal{T}(\Omega,\emptyset), \delta':Y\rightarrow \mathcal{T}(\Omega,\emptyset)$ such that $s'\sigma\delta = s$ and 
  $t'\sigma\delta' = t$ by Definition~\ref{def:classgrounding}. Let $\delta' = \{x_1\rightarrow s_1, ...,x_n\rightarrow s_n\}$. 
  Construct $\delta'' = \{x_1\tau\rightarrow s_1, ...,x_n\tau\rightarrow s_n\}$. 
  Now we can construct $\sigma'$ to be $\sigma\delta\delta''$ and we are done, since $\{\Gamma' \parallel s',\Delta'\parallel t'\tau\}\subseteq \text{norm}(A)$ and $s'\sigma'= s$ 
  and $t'\tau\sigma'= t$. 
  
  Now assume there exists a substitution $\sigma'$ such that 
  $\{\Gamma\sigma'\parallel s'\sigma',\Delta\sigma'\parallel t'\sigma'\}\subseteq \text{norm}(A)\sigma'$ and $s'\sigma'=s$ and 
  $t'\sigma'=t$.
  Define $\sigma:X\rightarrow \mathcal{T}(\Omega,\emptyset)$ and $\delta:(Y\cup Y\tau)\rightarrow \mathcal{T}(\Omega,\emptyset)$
  such that $\sigma'= \sigma\delta$. Let $\tau'$ be the reverse renaming of $\tau$. 
  Let $\delta = \{x_1\rightarrow s_1,...,x_n\rightarrow s_n, y_1\rightarrow t_1,...,y_m\rightarrow t_m\}$, where the $y_i$ are the 
  fresh variables introduced by $\tau$. Construct $\delta'=\{x_1\rightarrow s_1,...,x_n\rightarrow s_n\}$
  and $\delta''=\{y_1\tau'\rightarrow t_1,...,y_m\tau'\rightarrow t_m\}$. Now there must exist 
  $\{\Gamma\tau'\parallel s'\tau',\Delta\tau'\parallel t'\tau'\}\subseteq A$ such that 
  ($s'\tau'\sigma'\delta' = s$ or $s'\tau'\sigma'\delta'' = s$) and ($t'\tau'\sigma'\delta'' = t$ or $t'\tau'\sigma'\delta' = t$).
\end{proof}

\groundsound*
\begin{proof}
Let $Y'$ be the free variables of s. Choose $y'\in Y'$. There must exist a $\Gamma\parallel t\in A$ s.t. $y'\not\in vars(t)$.
By assumption $gnd_{\mathcal{M}}(E)\models s\sigma\delta\approx t\sigma\delta$. Now let 
$\delta''$ be $\delta$ but $y'$ maps to $\alpha$. Then $gnd_{\mathcal{M}}(E)\models s\sigma\delta\approx t\sigma\delta''$
since $t\sigma\delta=t\sigma\delta''$, and by hypothesis $gnd_{\mathcal{M}}(E)\models t\sigma\delta''\approx s\sigma\delta''$.
The constraints are still satisfiable since $\mathcal{M}$ is $\Pi$-closed. Now we can continue analogously for 
$s\sigma\delta''$ until we reach $s\sigma\delta'$ which shows the Lemma.
\end{proof}

\nonnormtonorm*
\begin{proof}
  Follows from Lemma~\ref{lem:normaltoground} and~\ref{lem:groundsound}.
\end{proof}

\cccsound*
\begin{proof}
It suffices to show this Lemma for all $\{\Gamma\parallel s,\Delta\parallel t\}\subseteq A$. Since constraints are always the same for each class in a run 
it follows 
for all $\{\Gamma\parallel s,\Delta\parallel t\}\subseteq \text{norm}(A)$ by Corollary~\ref{lem:nonnormtonorm}.
Proof by induction. Initially $\Pi$ is such that $\{s,t\parallel s, t\}\in \Pi$ for all $s\approx t\in E$. 
Now assume a grounding substitution $\sigma$ such that 
$s\sigma\in\mathcal{M}$ and $t\sigma\in\mathcal{M}$. Then $(s\approx t)\sigma\in \gnd_{\mathcal{M}}(E)$.
The initial $\{f_i(x_{1_i},\ldots,x_{k_i})\in\mathcal{M}\parallel f_i(x_{1_i},\ldots,x_{k_i})\}$ are single term classes.
So the assumption holds. Now assume that the assumption holds for $\Pi$ and we apply a rule:

\noindent
\textbf{1)} assume that Subsumption is applied. Then $B$ subsumes $A$. Thus for all $A'\in gnd(A)$ there 
        exists a $B'\in gnd(B)$ such that $A'\subseteq B'$.
        Thus we only remove redundant classes, so the assumption holds by i.h.
        
        \noindent
\textbf{2)} assume that $\mathit{Merge}$ is applied. Then there exists $\Gamma\parallel s\in A,\Delta\parallel t\in B$ and $\mu=mgu(s,t)$.
        We show that $\gnd_{\mathcal{M}}(E)\models s'\mu\sigma\approx t'\mu\sigma$ for all 
        $\{\Gamma'\mu \parallel s'\mu,\Delta'\mu\parallel t'\mu\}\subseteq (A'\cup B')\mu$ and
        grounding $\sigma$ such that $\Gamma'\mu\sigma$ and $\Delta'\mu\sigma$ satisfiable. 
        Let $\sigma' = \mu\sigma$. By i.h. $\gnd_{\mathcal{M}}(E)\models s'\sigma'\approx s\sigma'$
        for $\Gamma''\parallel s'\in \text{norm}(A)$, since $\Gamma''\subseteq\Gamma'$ and $\Gamma\subseteq\Gamma'$,
        and analogously for $t'\sigma'\approx t\sigma'$. Now we have $s'\sigma' \approx s\mu\sigma = t\mu\sigma \approx t'\sigma'$.
        Thus by transitivity of equality it holds $\gnd_{\mathcal{M}}(E)\models s'\mu\sigma\approx t'\mu\sigma$.

\noindent        
\textbf{3)} assume that $\mathit{Deduction}$ is applied. By i.h. $\gnd_{\mathcal{M}}(E)\models s'_i\sigma\approx t'_i\sigma$ for all
        grounding $\sigma$ such that $\Gamma_i\sigma$ and $\Delta_i\sigma$ are satisfiable and $1\leq i\leq n$. 
        Thus, by congruence of equality, $\gnd_{\mathcal{M}}(E)\models 
        f(s'_1,...,s'_n)\sigma\approx f(t'_1,...,t'_n)\sigma$ for all $\sigma$ such that $\Gamma'\sigma$ satisfiable,
        since $\Gamma' = \Gamma\land \Delta\land \Gamma_1\land...\land \Gamma_n \land \Delta_1\land...\land \Delta_n$.
        Thus $\gnd_{\mathcal{M}}(E)\models 
        f(s'_1,...,s'_n)\mu\sigma\approx f(t'_1,...,t'_n)\mu\sigma$ for all $\sigma$ such that $\Gamma'\mu\sigma$ satisfiable,
        since $f(s'_1,...,s'_n)\mu$ and $f(t'_1,...,t'_n)\mu$ are instances of $f(s'_1,...,s'_n)$ and $f(t'_1,...,t'_n)$.
      \end{proof}

\cccterm*
\begin{proof}
  By Lemma~\ref{lem:finnumclscomb} there are only finitely many possible $\mathcal{M}$-constrained classes that are not subsumed. Since all terms are constraint by $\mathcal{M}$ in 
  $\Rightarrow_{\CCVAR}$ and $\mathit{Merge}$ and $\mathit{Deduction}$ check if the new class is subsumed by another class, they can be applied only finitely often.
  A class removed by $\mathit{Subsumption}$ cannot be added again by $\mathit{Merge}$ or $\mathit{Deduction}$ since it is subsumed by another class in $\Pi$.
  Thus $\Rightarrow_{\CCVAR}$ is terminating.
\end{proof}
      
\ccccomp*
\begin{proof}
  We show by induction for any sequence $E_1\Rightarrow_{EQ} ...\Rightarrow_{EQ} E_n$ of applications of $\Rightarrow_{EQ}$ 
  with $E_1=\gnd_{\mathcal{M}}(E)$ there exists a sequence $\Pi_0\Rightarrow_{\CCVAR}...\Rightarrow_{\CCVAR} \Pi_m$ of applications of 
  $\Rightarrow_{\CCVAR}$ rules such that for all $s\approx t\in E_n$,
  where $s\in\mathcal{M}$ and $t\in\mathcal{M}$, 
  there exists a class $A\in \Pi_m$, $\Gamma\parallel l\in \text{norm}(A),\Delta\parallel r\in \text{norm}(A)$ and a 
  substitution $\sigma$ such that $l\sigma\approx r\sigma = s\approx t$, $\Gamma\sigma,\Delta\sigma$ satisfiable.
  Initially, this is true, since $\{s,t\parallel s, t\}\in \Pi_0$ for all $s\approx t\in E$.
  Since reasoning is based on ground terms only, we can ignore the $\mathit{Instance}$ rule of $\Rightarrow_{EQ}$.
  Now assume we are in step $i$.

  \noindent  1) Reflexivity. $E_i\Rightarrow_{EQ} E_i\cup \{t\approx t\}$, $t=f(s_1,\ldots,s_n)\in\mathcal{M}$. Then $t\in(\{f_i(x_{1_i},\ldots,x_{k_i})\in\mathcal{M}\parallel f_i(x_{1_i},\ldots,x_{k_i})\}\sigma)$
    for $\sigma = \{ x_{j_i}\mapsto s_j ~|~ 1\leq j\leq n\}$.

    \noindent  2) Symmetry. $E_i\cup \{t\approx t'\}\Rightarrow_{EQ} E_i\cup \{t\approx t', t'\approx t\}$, $t,t'\in\mathcal{M}$. 
    Then by i.h. there exists a class $A\in \Pi$ such that $\{l\sigma,r\sigma\}\subseteq \text{norm}(A)\sigma$ and $(l\approx r)\sigma= 
          t\approx t'$ for some substitution $\sigma$ because $\text{norm}(A)$ is normal. But then also $(r\approx l)\sigma= t'\approx t$.

          \noindent  3) Transitivity. $E_i = E'_i \cup \{s\approx t \land t\approx s'\}\Rightarrow_{EQ} E_i\cup \{s\approx s'\}=E_{i+1}$, $s,s',t\in\mathcal{M}$. 
          By hypothesis there must exist $\{A,B\}\subseteq \Pi$, 
          $\{\Gamma\parallel l,\Delta\parallel r\}\subseteq \text{norm}(A)$, $\{\Gamma'\parallel l',\Delta'\parallel r'\}\subseteq \text{norm}(B)$, 
          a grounding substitution $\sigma$, such that $(l\approx r)\sigma = s\approx t$ and 
          $(l'\approx r')\sigma = t\approx s'$ and $\Gamma\sigma, \Delta\sigma,\Gamma'\sigma,\Delta'\sigma$ all satisfiable. 
          If $(A'\cup B')\mu$ is subsumed by some $C\in \Pi$, then we 
          are already done, since there exists a $C'\in gnd(C)$ such that $\{s,s'\}\subseteq C'$.
          Otherwise $\mathit{Merge}$ is applicable and $(A'\cup B')\mu$ added to the state. 
          Then 
          $\{\Gamma\land \Delta\land \Gamma'\parallel l,\Delta\land \Gamma'\parallel r,
          \Delta\land \Gamma'\parallel l',\Delta'\land \Delta\land \Gamma'\parallel r'\}\mu\subseteq (A'\cup B')\mu$
          by the definition of Merge. Obviously, $(\Gamma\land \Delta\land \Gamma'\land \Delta')\mu$ is satisfiable since 
          $(\Gamma\land \Delta\land \Gamma'\land \Delta')\sigma$ is satisfiable. Thus there exists a $\sigma'$ such that 
          $\{s,s'\}\subseteq (A'\cup B')\mu\sigma'$.

          \noindent  4) Congruence. $E_i=E_i'\cup \{s_1\approx t_1,...,s_{m}\approx t_{m}\}\Rightarrow_{EQ} E_i\cup \{f(s_1,...,s_m)\approx f(t_1,...,t_m)\}=E_{i+1}$, 
          where $f(s_1,...,s_m), f(t_1,...,t_m)\in\mathcal{M}$. By hypothesis there must 
    exist $\{A_1,...,A_m\}\subseteq \Pi$, 
    $\{\Gamma_i\parallel l_i,\Delta_i\parallel r_i\}\subseteq \text{norm}(A_i)$ because $\text{norm}(A_i)$ is normal and a grounding substitution $\sigma$ such that $(l_i\approx r_i)\sigma = s_i\approx t_i$
    and $\Gamma_i\sigma,\Delta_i\sigma$ satisfiable. Now take two renamed copies of $\{f_i(x_{1_i},\ldots,x_{k_i})\in\mathcal{M}\parallel f_i(x_{1_i},\ldots,x_{k_i})\}$
    and grounding substitutions 
    $\sigma_1,\sigma_2$ such that $f(x'_1,...,x'_m)\sigma_1 = f(s_1,...,s_m)$ and $f(y'_1,...,y'_m)\sigma_2 = f(t_1,...,t_m)$ where the respective constraints are obviously satisfied. 
    Thus there must exist a 
    simultaneous most general unifier $\mu$ as defined in the Deduction rule. If $\{\Gamma'\mu\parallel f(l_1,...,l_m)\mu,\Gamma'\mu\parallel f(r_1,...,r_m)\mu\}$ is subsumed by some $C\in \Pi$, then we 
    are already done, since there exists a $C'\in gnd(C)$ such that $\{f(s_1,...,s_m), f(t_1,...,t_m)\}\subseteq C'$. 
    Otherwise Deduction is applicable and $\{\Gamma'\mu\parallel f(l_1,...,l_m)\mu,\Gamma'\mu\parallel f(r_1,...,r_m)\mu\}$ added to the state. 
    Since $\Gamma'\sigma$
    is satisfiable $\{f(s_1,...,s_m), f(t_1,...,t_m)\}\subseteq \{\Gamma'\mu\parallel f(l_1,...,l_m)\mu,\Gamma'\mu\parallel f(r_1,...,r_m)\mu\}\sigma$ has to hold.
\end{proof}

\dedstc*
\begin{proof}
  We build a unifier $\mu' = \{x_1\rightarrow s'_1,..., x_n\rightarrow s'_n, y_1\rightarrow t'_1,..., y_n\rightarrow t'_n\}$ for 
  the single term class and the renamed single term class  $\{f(y_1,...,y_n)\parallel f(y_1,...,y_n)\}$.
  Then $\Gamma''=f(x_1,...,x_n)\in\mathcal{M}\land f(y_1,...,y_n)\in\mathcal{M}\land \Gamma_1\land ...\land \Gamma_n\land \Delta_1\land ...\land 
  \Delta_n$. The resulting class is thus $\{\Gamma''\parallel f(s'_1,...,s'_n), f(t'_1,...,t'_n)\}\mu'$.
\end{proof}


We need the following auxiliary Lemma for the Proof of Lemma~\ref{lem:subsumptionmatch}.

\begin{lemma}\label{lem:subsummatching}
  Let $A, B$ be classes. If $B$ subsumes $A$ by matching then $B$ subsumes $A'$ by matching 
  for all $A'\in gnd(A)$.
\end{lemma}
\begin{proof}

  Assume $B$ subsumes $A$ by matching. By assumption there exists a 
  substitution $\sigma$ such that for all $\Gamma\parallel t\in A$ there exists a $(\Delta\parallel s)\sigma\in B\sigma$ and $\tau$ such
  that $t=s\sigma\tau$ and $\forall\delta. (\Gamma\delta\rightarrow\exists \delta'.\Delta\sigma\tau\delta\delta')$. $\sigma$ matches all separating variables of $B$ to terms containing only separating variables
  of $A$. If not, then every $t\in A$ contains a free variable in $cdom(\sigma)$. But then these are separating variables.
  Contradiction.
  
  Let $A'\in gnd'(A)$ and $\mu:X\rightarrow 
  \mathcal{T}(\Omega)$ be the substitution such that $A\mu = A'$, where $X$ are the separating variables of $A$. 
  Now construct substitution $\sigma' = \sigma\mu$. Then
  for all $(\Gamma\parallel t)\mu\in A'$ there exists a $(\Delta\parallel s)\sigma'\in B\sigma'$ and $\tau' = \tau\mu$ such
  that $t\mu=s\sigma'\tau'$ and $\forall\delta. (\Gamma\mu\delta\rightarrow\exists \delta'.\Delta\sigma'\tau'\delta\delta')$,
  since $s\sigma'\tau'=s\sigma\mu\tau\mu = s\sigma\tau\mu$.

  Now let $A''\in gnd(A')$. We have $gnd(A') = \{A''\}$. Then there exist grounding substitutions $\sigma_1,...,\sigma_n$ such that
  $A'' = A'\sigma_1\cup...\cup A'\sigma_n$. Now, construct $\sigma'' = \sigma'$. Then
  for all $(\Gamma \parallel t)\mu\sigma_i\in A''$ there exists a $(\Delta\parallel s)\sigma''\in B\sigma''$ and $\tau'' = \tau'\sigma_i$ such
  that $t\mu\sigma_i=s\sigma''\tau''$ and $\forall\delta. (\Gamma\mu\sigma_i\delta\rightarrow\exists \delta'.\Delta\sigma''\tau''\delta\delta')$,
  since $s\sigma''\tau''=s\sigma\mu\tau\mu\sigma_i = s\sigma\tau\mu\sigma_i$.

\end{proof}

\subsumptionmatch*
\begin{proof}

Assume that $B$ does not subsume $A$. Then there exists an $A'\in gnd'(A)$ and 
an $A''\in gnd(A')$ such that there exists no $B'\in gnd(B)$ such that $A''\subseteq B'$.
Now assume that 
there exists a $\sigma$ such that for any $(\Gamma\parallel t)\tau'\in A'', \Gamma\parallel t\in A'$ there is a 
$(\Delta\parallel s)\sigma\in B\sigma$ and $\tau$ with $s\sigma\tau = t\tau'$ and
$\forall\delta. (\Gamma\tau'\delta\rightarrow\exists \delta'.\Delta\sigma\tau\delta\delta')$.
So there exist $\tau_1,...,\tau_n$ such that 
$A''\subseteq B\sigma\tau_1\cup...\cup B\sigma\tau_n$.
Obviously, the free variables of $B$ are also free variables of $B\sigma$.
$B\sigma$ has no separating variables, otherwise $B\sigma\tau_i$ would not be ground
for all $1\leq i \leq n$. Thus $gnd(B\sigma) = \{B''\}$ and 
$B\sigma\tau_1\cup...\cup B\sigma\tau_n \subseteq B''$. Thus, $A''\subseteq B''$, contradicting 
assumption.
Thus, by Lemma~\ref{lem:subsummatching} $B$ cannot subsume $A$ by matching.

\end{proof}

\reduceconstraints*
\begin{proof}
  A constraint is satisfiable with respect to the symbol counting order, if it is
  satisfiable by substituting a constant for all variables. Concerning the semantics
  of classes, variables only occurring in the constraint do not play a role as long
  as the constraint is satisfied. Thus $\gnd(A) = \gnd(A\sigma)$.
\end{proof}

\corcoabstr*
\begin{proof}
  by applying the definitions
\end{proof}

\cotcoelim*
\begin{proof}
Follows from Lemma~\ref{lem:corcoabstr} and the above observation that the quantifier alternation can be removed.
\end{proof}

\end{document}